\documentclass[journal,english,a4paper]{IEEEtran}
\usepackage[T1]{fontenc}
\usepackage[latin1]{inputenc}
\usepackage{graphicx}
\usepackage{amssymb}
\usepackage{amsmath}
\makeatletter

\providecommand{\boldsymbol}[1]{\mbox{\boldmath $#1$}}
\newcommand{\ZZ}{\mathbb{Z}}
\newcommand{\RR}{\mathbb{R}}
\newcommand{\CC}{\mathbb{C}}
\newcommand{\SC}{\mathcal{S}}
\newcommand{\rcomp}{R_{\mathrm{comp}}}

\newtheorem{theorem}{Theorem}

\newtheorem{proposition}{Proposition}
\usepackage{babel}
\makeatother

\begin{document}

\title{The Compute-and-Forward Protocol: Implementation and Practical Aspects}
\author{Ali Osmane and Jean-Claude Belfiore\\
TELECOM ParisTech, Paris, France\\
Email:
\texttt{\{osmane,belfiore\}@telecom-paristech.fr}}

\maketitle
\begin{abstract}
In a recent work, Nazer and Gastpar proposed the Compute-and-Forward strategy
as a physical-layer network coding scheme. They described a code 
structure based on nested lattices whose algebraic structure makes the scheme reliable and
efficient. In this work, we consider the implementation of their scheme for real Gaussian channels
and one dimensional lattices. 
We relate the maximization of the transmission rate to the lattice shortest vector problem. We explicit, 
in this case, the maximum likelihood criterion and show that it can be implemented by using an 
Inhomogeneous Diophantine Approximation algorithm.
\end{abstract}

\section{Introduction}
In \cite{PNC_ZLL}, Zhang \textit{et al.} introduced the physical-layer network
coding concept (PNC) in order to turn the broadcast property of the wireless
channel into a capacity boosting advantage. Instead of considering the
interference as a nuisance, each relay converts an interfering signal into a
combination of simultaneously transmitted codewords. PNC concept has received a
particular interest in the last years because it provides means of embracing
interference and improving network capacity.

In a recent work \cite{CF_Bob_Gas}, Nazer and Gastpar proposed a new
physical-layer network coding scheme. Their strategy, called
compute-and-forward (CF), exploits interference to obtain higher end-to-end
transmission rates between users in a network. The relays are required to decode
noiseless linear equations of the transmitted messages using the noisy linear
combination provided by the channel. The destination, given enough linear
combinations, can solve the linear system for its desired messages. This strategy is based on the
use of structured codes, particularly nested lattice codes to ensure that
integer combinations of codewords are themselves codewords. The authors
demonstrated its asymptotic gain using information-theoretic tools.

The authors in \cite{CF_FSK} followed the framework of Nazer and Gastpar and
showed the potential of the compute-and-forward protocol using an algebraic
approach. They related the Nazer-Gastaper's approach to the theorem of finitely
generated modules over a principle ideal domain (PID). They gave sufficient
condition for lattice partitions to have a vector space structure which is a
desirable property to make them well suited for physical-layer network coding.
Then, they generalized the code construction and developed encoding and decoding
methods.

In \cite{CF_DF_NW}, the authors proved that the lattice implementation of
compute-and-forward as proposed by Nazer and Gastpar suffers from a loss in
number of achieved degrees of freedom. They proposed a different implementation
consisting of a modulation scheme and an outer code and showed that it achieves
full degrees of freedom as if full cooperation among transmitters and among
relays was permitted. In their scheme, the channel coefficients are known
throughout the network. In \cite{CF_MCS_HN}, the authors designed a
modulation/coding scheme inspired by the compute-and-forward protocol for the
wireless two-way relaying channel.

In this work, we consider the practical aspects of the compute-and-forward
protocol. We implement the protocol described by Nazer and Gastpar. We
explain how to obtain the integer coefficients that maximize the rate. We also
propose a decoding technique based on maximum likelihood. Finally,
we show some simulation results. All the practical aspects are demonstrated here for one-dimensional 
real constellations. 

\begin{figure}
\centering
\includegraphics[width=0.55\linewidth]{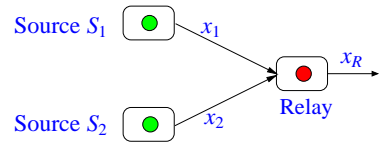}
\caption{System model: 2 sources and one relay.}
\label{fig:Sys_Model}
\end{figure}
%===============================================================================
% ===================================
%	
%===============================================================================
% ===================================

\section{System Model and Assumptions}
In our model, we consider one relay receiving messages from two sources $S_1$
and $S_2$ and transmitting a linear combination of these two messages, as described in Figure \ref{fig:Sys_Model}. 
The relay observes a noisy linear combination of the transmitted
signals through the channel. Received signal at the relay is expressed as, 
\begin{equation}
 y = h_1 x_1 + h_2 x_2 + z.
\end{equation}
The relay searches for the integer coefficient vector $\boldsymbol{a}= [a_1
\hspace{0.1cm} a_2]^T$ that maximizes the transmission rate. It then decodes a
noiseless linear combination of the transmitted signals, 
\begin{equation}
 x_R = a_1 x_1 +a_2 x_2, 
\end{equation}
and retransmits it
to the destination or another relay. 
We consider a real-valued channel model with real inputs and outputs. 
The channel coefficients $h_1$ and $h_2$ are real, i.i.d. Gaussian, 
$h_i \sim \mathcal{N}(0, 1)$. $z$ is Gaussian, zero mean, with variance $\sigma^2=1$ ($ z \sim \mathcal{N}(0,1)$). 
Let $\boldsymbol{h}= [h_1
\hspace{0.1cm} h_2]^T$ denotes the vector of channel coefficients. Source
symbols $x_i$ are integers and verify $\lvert x_i \rvert \leq s_m$, i.e.,
$x_i \in \SC=\left\lbrace -s_m, -s_m+1, \hdots, s_m \right\rbrace $. $S_1$ and
$S_2$ transmit $x_1$ and $x_2$, respectively. Both sources have no channel side
information (CSI). CSI is only available at the relay.

%===============================================================================
% ===================================
%	
%===============================================================================
% ===================================

\section{Compute-and-Forward}
In what follows, we use the expression of the computation rate $\rcomp$ given by Nazer and Gastpar \cite{CF_Bob_Gas} in 
order to find a vector $\boldsymbol{a}$ maximizing it. We show that the maximization of $\rcomp$ is 
equivalent to the search of a shortest vector in a lattice. Then, based on the likelihood expression,
we show that decoding is equivalent to an Inhomogeneous Diophantine
Approximation.

\subsection{Achievable Computation Rate}
The primary goal of the decode-and-forward is to enable higher achievable rates
across the network. Nazer and Gastpar showed that the relays can recover any set
of linear equations with coefficient vector $\boldsymbol{a}$ as long as the
message rates are less than the computation rate
\begin{equation} \label{comp_rate}
 \rcomp(\boldsymbol{h}, \boldsymbol{a}) = \log \left( \left( \lVert \boldsymbol{a}
\rVert^2 - \dfrac{\mathsf{SNR}\lvert \boldsymbol{h}^\dagger
\boldsymbol{a}\rvert^2}{1 + \mathsf{SNR} \lVert \boldsymbol{h}
\rVert^2}\right)^{-1} \right) 
\end{equation}
where this rate is achievable by scaling the received signal by the MMSE coefficient \cite{CF_Bob_Gas}. 
We are interested in finding the coefficient vector with the highest computation
rate. This is given in the following theorem. The result is obtained for a
relay combining $N$ symbols and for complex-valued channels. 
% % $\vspace{0.3cm}$
\begin{theorem}\label{theorem-a}
 For a given $\boldsymbol{h} \in \CC^N$ (resp. $\RR^N$), $\rcomp(\boldsymbol{h}, \boldsymbol{a})$ 
is maximized by choosing $\boldsymbol{a} \in \ZZ[i]^N$ (resp. $\ZZ^N$) as
\begin{equation} \label{max_rate_th}
 \boldsymbol{a} = \arg \min \limits_{\boldsymbol{a}\neq\boldsymbol{0}} \left(\boldsymbol{a}^{\dagger}
\boldsymbol{G} \boldsymbol{a} \right )
\end{equation}
where
\begin{equation}
 \boldsymbol{G}= \boldsymbol{I} - \dfrac{\mathsf{SNR}}{1
+ \mathsf{SNR} \lVert \boldsymbol{h} \rVert^2} \boldsymbol{H}.
\end{equation}
$\boldsymbol{H} = [H_{ij}]$, $H_{ij}=h_i h_j^{\star}$, $1 \leq i, j \leq N$ and $\dagger$ is for the Hermitian transpose 
(resp. the regular transpose).
% $\vspace{0.3cm}$
\end{theorem}

\begin{proof}
Maximizing $\rcomp(\boldsymbol{h}, \boldsymbol{a})$ is equivalent to the following
minimization
\begin{equation} \label{max_rate}
 \min \limits_{\boldsymbol{a}\neq \boldsymbol{0}} \left\lbrace \lVert \boldsymbol{a} \rVert^2 +
\mathsf{SNR} \lVert \boldsymbol{h} \rVert^2 \lVert \boldsymbol{a} \rVert^2 -
\mathsf{SNR} \lvert\boldsymbol{h}^{\dagger} \boldsymbol{a}\rvert^2\right\rbrace .
\end{equation}
We can write
\begin{equation} 
 \lvert \boldsymbol{h}^\dagger \boldsymbol{a}\rvert^2 =\sum
\limits_{i,j} h_i h_j^{\star} a_i^{\star} a_j 
\end{equation}
As $\boldsymbol{H}=[H_{ij}]$, $H_{ij}=h_ih_j^{\star}$, $1 \leq i, j
\leq N$, it follows that $\sum
\limits_{i,j} h_i h_j^{\star} a_i^{\star} a_j=
\boldsymbol{a}^\dagger \boldsymbol{H} \boldsymbol{a}$. Using these notations, we
can write (\ref{max_rate}) as
\begin{equation}
 (1 + \mathsf{SNR} \lVert \boldsymbol{h} \rVert^2) \min \limits_{\boldsymbol{a}\neq \boldsymbol{0}} 
\boldsymbol{a}^\dagger \left[ \boldsymbol{I} - \dfrac{\mathsf{SNR}}{1 +
\mathsf{SNR} \lVert \boldsymbol{h} \rVert^2} \boldsymbol{H} \right]
\boldsymbol{a}.
\end{equation}
$\boldsymbol{I} - \dfrac{\mathsf{SNR}}{1 + \mathsf{SNR} \lVert
\boldsymbol{h} \rVert^2} \boldsymbol{H} $ has $N$ strictly positive eigenvalues. It is
then positive definite. Now, the problem is reduced to the minimization of
$\boldsymbol{a}^\dagger \boldsymbol{G}\boldsymbol{a}$.
\end{proof}

\begin{proposition}
 Searching for the vector $\boldsymbol{a}$ that minimizes Equation (\ref{max_rate_th}) of theorem \ref{theorem-a}
is equivalent to a ``Shortest Vector'' problem for the lattice $\Lambda$ whose Gram matrix is $\boldsymbol{G}$. 
% $\vspace{0.3cm}$
\end{proposition}

\begin{proof}
As $\boldsymbol{G}$ is a definite positive hermitian (resp. symmetric) matrix, it is the Gram matrix of a lattice $\Lambda$.
This lattice is either a $\ZZ[i]-$ lattice in the complex case, or a $\ZZ -$ lattice in the real case. 
Then, the minimization problem in theorem \ref{theorem-a} is equivalent to find a non zero vector in $\Lambda$ with shortest length. 
\end{proof}

Algorithms for solving this problem are given in \cite{Algo_Short_vect}. The
best known one is the Fincke-Pohst algorithm \cite{Algo_Fin_Poh}.

\subsection{Recovering Linear Equations}
The relay aims to decode a linear equation of the transmitted messages
and passes it to the destination or another relay. After calculating the vector
$\boldsymbol{a}$ as in (\ref{max_rate_th}), the
relay recovers a linear combination of the transmitted signal $x_1$ and $x_2$.
We rewrite the received signal at the relay in the following form
\begin{equation}\label{in_out}
 y = \lambda + \xi_1 x_1 + \xi_2 x_2 + z
\end{equation}
where $\lambda$ is an integer, $\xi_i = h_i - a_i$ and $z$ is the additive white
noise.
The recovered linear equation $\lambda = a_1 x_1 + a_2 x_2$ is a linear
Diophantine equation. This equation admits the following solutions.

\subsection{Solution of the Linear Diophantine Equation}
If $\lambda$ is a multiple of the greatest common divisor (gcd) of $a_1$ and $a_2$,
then the Diophantine equation has an infinite number of solutions. The
\textit{Extended Euclid Algorithm} allows to exhibit a particular solution
$(u_1,u_2)$ to $a_1x_1+a_2x_2=g$ \cite{Algo_CLRS}. 
The set of all solutions is obtained as follows
\begin{equation}\label{dioph_sol}
\begin{cases}
 x_1=\frac{u_1}{g}\lambda + \frac{a_2}{g}k \\
 x_2=\frac{u_2}{g}\lambda - \frac{a_1}{g}k
\end{cases} 
\end{equation}
$g=a_1\wedge a_2$ is the gcd of $a_1$ and $a_2$, $k \in \mathbb{Z}$.

\subsection{Decoding Metric}
The Maximum Likelihood decoder maximizes $p(y/\lambda)$ over all
possible values of $\lambda$.  The conditional
probability $p(y/\lambda)$ can be expressed as,
\begin{equation}\label{prob_error_1}
 p(y/\lambda) = \sum_{\underset{a_1 x_1 + a_2 x_2 = \lambda}{(x_1, x_2)}}
p(y/x_1,x_2) p(x_1, x_2)
\end{equation}
where 
\begin{equation}\label{prob_error_2}
 p(y/x_1, x_2) \propto \exp \left[ -\frac{(y - h_1 x_1 - h_2 x_2)^2}{2 \sigma^2}
\right] 
\end{equation}
and $x_1$, $x_2$ are (\textit{a priori}) equiprobable and given by (\ref{dioph_sol}). The decoding rule is now to find, 
\begin{equation}\label{ML}
\hat{\lambda}=\arg \max \limits_{\lambda}\varrho(\lambda):= \sum \limits_{k=-\infty}^{+\infty} \exp
\left[ -\frac{(y - \beta \lambda + k \alpha)^2}{2 \sigma^2} \right]
\end{equation}
where $\beta = \frac{1}{g}\left( h_1u_1+h_2u_2\right)$, $\alpha = \frac{1}{g}\left(
h_2a_1-h_1a_2\right)$.

In \cite{GL_MR}, it has been proved that, for $\lambda\in\RR$, 
$\varrho(\lambda)$ achieves its maximum for
\[
\lambda\in \frac{\alpha}{\beta} \ZZ + \frac{y}{\beta}, 
\]
i.e. for all values of $\lambda$ such that $y - \beta \lambda + k \alpha=0$. 
Since we want to maximize $\varrho(\lambda)$ for $\lambda\in\ZZ$, the solution is given 
by the integer-valued couple $(\lambda,k)$ minimizing $\left|y-\beta \lambda + k \alpha\right|$. 
Thus, since $x_1,x_2\in\mathcal{S}$ which is a finite subset of $\ZZ$ and verify Equation (\ref{dioph_sol}), we state a new minimization 
problem which is equivalent to (\ref{ML}),
\begin{equation}
 \hat{\lambda}=\arg \limits_{\lambda} \min_{\underset{x_1,x_2\in\mathcal{S}}{(\lambda,k)}}\left|y-\beta \lambda + k \alpha\right|.
\end{equation}
The problem is therefore equivalent to the minimization of 
\begin{equation}\label{IDA}
 F(k,\lambda) =  \left| k \alpha' - \lambda + y'\right|
\end{equation}
$\alpha'=\alpha/\beta$ and $y'=y/\beta$. 
The minimization is called
\textit{Inhomogeneous Diophantine Approximation} in the absolute sense. It consists of finding the best approximation of a real number $\alpha'$ by
a rational number $\lambda/k$, $k$ $\in$ $\mathbb{N}$, given an additional real
shift $y'$, while keeping the denominator $k$ as small as possible. In the general settings for such problems,
an error approximation function $F(k,\lambda)$ is set and it is stated 
that a rational number $\lambda/k$ is the Best Diophantine Approximation if, for all other rational numbers $\lambda'/k'$ 
\begin{equation}
 k' \leq k \Rightarrow  F(k',\lambda') \geq F(k,\lambda).
\end{equation}
In our case, in addition to the error approximation function, limits are
imposed by the finite constellation $\SC$ to which the transmitted symbols belong. The algorithms used to
find the best Diophantine approximations of real numbers are in general simple
and easy to implement. The best known one is the Cassel's algorithm
\cite{IDA_Cassels}. In \cite{CLark_Th}, the authors develop and compare several
ones.
\vspace{-0.1cm}
\section{Numerical Results}
In the simulations, the set of symbols is of the form $\mathcal{S} = \left\lbrace
-s_m, \hdots, s_m \right\rbrace $. We consider two sources transmitting $x_1$
and $x_2$, and one relay recovering a linear equation of $x_1$ and $x_2$ with
integer coefficients.

At first, based on its CSI, the relay finds the vector
$\boldsymbol{a}$ as the shortest vector described in theorem \ref{theorem-a}.
Then, the relay finds a particular solution of the linear Diophantine equation
$a_1 x_1 + a_2 x_2 = g$ using the \textit{Extended Euclid} algorithm. Finally,
the relay searches for the couple $(k, \lambda)$ which gives the
best inhomogeneous Diophantine approximation by minimizing the function $F$ defined in (\ref{IDA}).
\begin{figure}[ht]
\centering
\includegraphics[width=0.52\linewidth, angle=0]{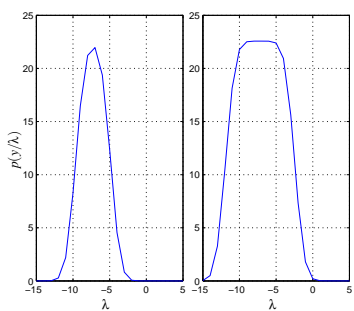}
\vspace{-0.3cm}
\caption{$p(y/\lambda)$ for $\boldsymbol{h}=[-1.274 \hspace{0.1cm} 0.602]^T$, $\boldsymbol{a}=[2 \hspace{0.1cm} -1]^T$, 
$\mathsf{SNR}=40dB$, $x_1=-2$ and $x_2=3$. $p(y/\lambda)$ is maximized for one value, $\lambda=-7$ in the left subfigure 
while it is maximized for several values of $\lambda$ in the right one.}
\vspace{-0.2cm}
\label{fig:prob_func}
\end{figure}

In Figure \ref{fig:IDA_vs_Dec}, we show the error probability of our system for
three different constellations $\SC$, defined by $s_m=5, 7, 10$, respectively. For $s_m=5$ or
less, the diversity order of the system is 1 for real entries (which would correspond to a diversity order equal to 2 with complex symbols). 
For $s_m > 6$, the diversity order collapses to $1/2$. 
\begin{figure}[ht]
\centering
\includegraphics[width=0.65\linewidth, angle=270]{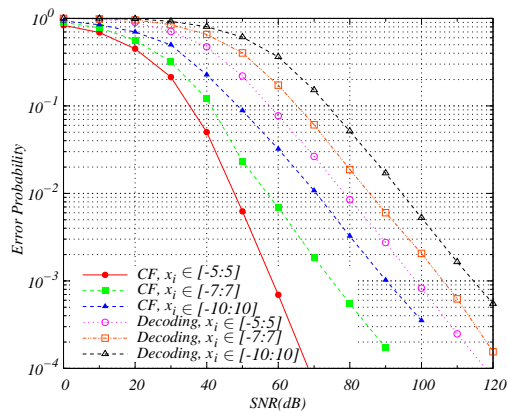}
% \vspace{-0.5cm}
\caption{Error Probability using the Inhomogeneous Diophantine Approximation
versus decoding both symbols $x_1$ and $x_2$.}
\label{fig:IDA_vs_Dec}
\end{figure}
This is due to the fact that $p(y/\lambda)$ is constant, as a function of $\lambda$, on a bigger interval
giving rise to ambiguities as shown in Figure \ref{fig:prob_func}. Still in 
Figure \ref{fig:IDA_vs_Dec}, we plotted the error probability for when the relay decodes both symbols $x_1$ and $x_2$.
The diversity order in this case is $1/2$ for all values of $s_m$.
For the case of complex-valued channels and symbols, we expect a doubled value of all the
diversity orders.

\section{Conclusion}
In this paper, we considered the Compute-and-Forward scheme with real-valued channels. 
We provided a method for maximizing the transmission rate and developed a decoding 
strategy. Numerical results showed the performance of our decoding method. 
We believe that it is a first step towards a rich and fruitful multidimensional approach. 
% \section*{Acknowledgment}
% The authors wish to thank Sheng Yang for interesting discussions. 
% \vspace{-0.3cm}

\end{document}